\documentclass{article}
\usepackage[utf8]{inputenc}

\usepackage{graphicx}
\usepackage{comment}
\usepackage{enumitem}
\usepackage{tikz}
\usepackage{array}
\usepackage{pgfplots}
\usetikzlibrary{pgfplots.fillbetween}
\usepackage{mathtools}

\usepackage{graphicx}

\usepackage{amsmath}
\usepackage{amsfonts}
\usepackage{amssymb}
\usepackage{amsthm}
\usepackage{hyperref}

\usetikzlibrary{shapes,backgrounds}
\usetikzlibrary{shapes.misc, positioning}
\usepackage{pgfplots}
\usepgfplotslibrary{fillbetween}
\usetikzlibrary{intersections}
\pgfdeclarelayer{ft}
\pgfsetlayers{main,ft}
\usetikzlibrary{positioning, fit}

\usetikzlibrary {positioning}
\tikzset{main node/.style={circle,fill=blue!40,draw,minimum size=2cm,inner sep=0pt},}

\usepackage{epsfig}
\usepackage{graphicx}

\usepackage{pst-node}
\usepackage{tikz-cd}
\pgfplotsset{compat=1.17}

\title{Representing preorders with injective monotones}
\author{Pedro Hack, Daniel A. Braun, Sebastian Gottwald}
\date{ }

\newtheorem{proposition}{Proposition}

\newtheorem{lemma}{Lemma}

\theoremstyle{definition}
\newtheorem{definition}{Definition}

\begin{document}

\maketitle

\begin{abstract}
We introduce a new class of real-valued monotones in preordered spaces, injective monotones. We show that the class of preorders for which they exist lies in between the class of preorders with strict monotones and preorders with countable multi-utilities, improving upon the known classification of preordered spaces through real-valued monotones. We extend several well-known results for strict monotones (Richter-Peleg functions) to injective monotones, we provide a construction of injective monotones from countable multi-utilities, and relate injective monotones to classic results concerning Debreu denseness and order separability. Along the way, we connect our results to Shannon entropy and the uncertainty preorder, obtaining new insights into how they are related. In particular, we show how injective montones can be used to generalize some appealing properties of Jaynes' maximum entropy principle, which is considered a basis for statistical inference and serves as a justification for many regularization techniques that appear throughout machine learning and decision theory.
\end{abstract}

\section{Introduction}
\label{intro}

The set of all preordered spaces $(X,\preceq)$ is structured according to how well their preorder can be represented by real-valued \emph{monotones}, that is, functions $u: X \to \mathbb{R}$ such that $x \preceq y$ implies $f(x) \leq f(y)$ $\forall x,y \in X$ \cite{evren2011multi,ok2002utility}. Two major classification methods can be distinguished depending on whether one considers a single monotone \cite{alcantud2016richter} or a whole family $U$ of monotones encapsulating all the information in $\preceq$, called a \emph{multi-utility} \cite{evren2011multi}. More precisely, if $U$ is a multi-utility for $(X,\preceq)$ then $\forall x,y \in X$ we have $x \preceq y$ if and only if $u(x) \leq u(y)$ $\forall u \in U$. Without further contraints, monotones and multi-utilities are, however, not very useful from a classification perspective as they exist for any preordered space. They become more useful when adding constraints. For example, there are preorderd spaces without \emph{strict} monotones, that is, without monotones $u$ such that, $u(x) <u(y)$ whenever\footnote{Here, $x\prec y$ means $x \preceq y$ and $\neg(y \preceq x)$} $x\prec y$. Strict monotones, also known as \emph{Richter-Peleg} functions, have been extensively studied \cite{alcantud2013representations,alcantud2016richter,richter1966revealed,peleg1970utility} and are related to other features of the preorder such as its maximal elements. In the case of multi-utilities, the cardinality is an important property for the classification of preordered spaces, with \emph{countable} multi-utilities playing a central role   \cite{bevilacqua2018representation}. Of particular importance are \emph{utility functions} \cite{debreu1954representation,debreu1964continuity}, that is, multi-utilities consisting of a single function\footnote{Notice, having a utility function implies $\preceq$ is total, that is, any pair of points $x,y \in X$ can be related by $\preceq$.}.

Here, we introduce \emph{injective monotones}, which are monotones $u$ such that $u(x)=u(y)$ implies both $x \preceq y$ and $y \preceq x$. Preorders for which they exist form a category between preorders with strict monotones and preorders with countable multi-utilities, as we show in Proposition \ref{monotones different}, Proposition \ref{countable implies injective monotone} and Proposition \ref{inj mono no count mu}. Hence we improve on the existing classification of preorders by adding a new distinct class. More precisely, in Section \ref{injective monotones} we define injective monotones and prove some simple properties. After discussing their relation to optimization in Section \ref{monotones and optimization}, we take a look at the role of multi-utilities in Section \ref{monot and multi-ut}, in particular we construct injective monotones from countable multi-utilities and show that the converse does not hold. Finally, in Section \ref{monot and denseness}, we consider separability properties of preorders that are sufficient for the existence of strict and injective monotones, introducing a new notion of Debreu separability, that allows to extend previous results on strict monotones to corresponding analogues for injective monotones.

In the following section, we introduce our running example to which we come back several times throughout the development of the general theory. In particular, we discuss the relation between the uncertainty preorder from majorization theory \cite{arnold2018majorization}, which has Shannon entropy as a strict monotone, and the \emph{maximum entropy principle} that appears in many different parts of science.

\section{Example: the uncertainty preorder and Shannon entropy}
\label{example}

The outcome of a random variable with a narrow probability distribution is easier to predict than the outcome of a random variable with a less concentrated distribution. For example, the result of throwing an unbalanced coin is easier to predict than the one of a balanced coin. In other words, a wider distribution contains more \emph{uncertainty} than a narrower distribution. This idea is captured by a binary relation on the space $\mathbb P_\Omega$ of probability distributions on a set $\Omega$: the uncertainty preorder $\preceq_U$, defined for finite $\Omega$ by 
\begin{equation}
\label{uncert rela}
    p \preceq_U q \ \iff \ u_i(p) \leq u_i(q) \ \  \forall i\in \{1,..,|\Omega|-1\} \, ,
\end{equation}
where $u_i(p) \coloneqq -\sum_{n=1}^{i} p_n^{\downarrow}$ and $p^{\downarrow}$ denotes the decreasing rearrangement of $p$ (same components as $p$ but ordered decreasingly). Notice, $\preceq_U$ is known in mathematics, economics, and quantum physics as \emph{majorization} \cite{hardy1952inequalities,marshall1979inequalities,arnold2018majorization,brandao2015second}, originally developed by Lorenz  \cite{lorenz1905methods} and Dalton \cite{dalton1920measurement} among others, to measure wealth and income inequality. An intuitive way to think of $p \preceq_U q$ is that $q$ is the result of finitely many transfers of pieces of probability from a more likely to a less likely option in $p$ \cite{gottwald2019bounded}. In other words, $q$ is more spread out or less biased, and thus, contains more uncertainty than $p$. For instance, a Dirac distribution is the smallest, and the uniform distribution is the largest, with respect to $\preceq_U$, among all distributions on $\Omega$.

There is, however, a downside to this intuitive notion of uncertainty: what if $p$ and $q$ do not have this relationship? For example, if $p=(0.6,0.2,0.2,0,..,0)$ and $q=(0.5,0.4,0.1,0,..,0)$, then $p$ and $q$ cannot be related by $\preceq_U$. Instead, the most common way to measure uncertainty is to use an entropy functional, such as the Shannon entropy, $H(p)\coloneqq -\mathbb E_p[\log p]$, or one of various alternative entropy proposals, including Renyi entropy \cite{renyi1961measures}, Tsallis entropy \cite{tsallis1988possible}, and many more \cite{csiszar2008axiomatic}. Even though, in general, $\preceq_U$ cannot be fully represented by any of these so-called generalized entropies $F$ \footnote{Since $\preceq_U$ is not total for any $|\Omega|>2$, it has no utility function}, it is noteworthy that all of them are \emph{monotones} with respect to $\preceq_U$. While the converse is not true for any \emph{single} $F$, there are collections $\mathcal F$ which constitute a multi-utility, e.g.~in the case of finite $\Omega$, $\mathcal F = \{\sum_{n=1}^{|\Omega|} f(p_n) \, | \, f \text{ concave} \}$ \cite{schur1923uber}, or even $\mathcal F = \{u_i\}_{i=1}^{|\Omega|-1}$ by the definition of $\preceq_U$ \eqref{uncert rela}.

The preference towards unbiased distributions, that is represented by any monotone of $\preceq_U$, is of particular relevance in the \emph{maximum entropy principle}, where (Shannon) entropy serves as a counter-acting force against the bias towards the maximal elements of a given ``energy'' function $E$. Going back to the \emph{principle of insufficient reason} \cite{bernoulli1713ars}, today the maximum entropy principle appears in virtually all branches of science. For example, it is often used as a general principle to explain the raison d'etre behind all kinds of ``soft'' versions of known machine learning methods, especially in reinforcement learning \cite{williams1991function,fox2016taming}, but also in models of robust and resource-aware decision making \cite{maccheroni2006ambiguity,still2009information,tishby2011information,ortega2013thermodynamics}. Basically, whenever there appears a trade-off between precision and uncertainty, there is a good chance that the maximum entropy principle is applied \cite{gottwald2020two}. 

The underlying goal of the maximum entropy principle is to select a \emph{typical} distribution among a set of candidate distributions satisfying a given constraint, usually of the form $\langle E\rangle_p = c$, where $\langle E\rangle_p$ denotes the expectation of a random variable $E$ with respect to the probability measure $p$. In Wallis' derivation of the maximum entropy principle, typicality is measured by the number of possibilities of assigning $n$ elements among $N$ groups, under the limit of infinitely many elements ($n\to\infty$) such that the statistical probabilities  $p_i = \frac{n_i}{N}$ of belonging to a specific group $i$ remain finite ($n_i$ denotes the number of elements in group $i$) \cite{jaynes2003probability}. However, we can also think of typicality as containing the least amount of bias, or in other words, the maximal amount of uncertainty. Thus, when considering the uncertainty preorder $\preceq_U$ as the most basic way to decide about the difference in uncertainty between two distributions, then the ultimate goal of the maximum entropy principle becomes to obtain the maximal elements of $\preceq_U$, inside the given constraint set.

Even though, generally, we are not guaranteed to find all maximal elements of $\preceq_U$ when maximizing entropy, maximum entropy solutions are in fact maximal elements of $\preceq_U$, as entropy is a strict monotone. Furthermore, since the maximum entropy principle maximizes a strictly concave functional $H$ over a convex subset, it yields a unique maximal element of $\preceq_U$. In contrast, injective monotones, which exist for $\preceq_U$ (see Proposition \ref{countable implies injective monotone}), preserve this uniqueness property up to equivalence (see Proposition \ref{optimization charac of R-P}), without asking for the additional structural requirements of concavity.

\section{Injective monotones} \label{injective monotones}

A \emph{preorder} $\preceq$ on a set $X$ is a reflexive ($x \preceq x$ $\forall x \in X$) and transitive ($x \preceq y$ and $ y \preceq z$ implies $x \preceq z$ $\forall x,y,z \in X$) binary relation. A tuple $(X, \preceq)$ is called a \emph{preordered space}. An antisymmetric ($x \preceq y$ and $y \preceq x$ imply $x=y$ $\forall x,y \in X$) preorder $\preceq$ is called a \emph{partial order}. The relation $x \sim y$, defined by $x\preceq y$ and $y\preceq x$, forms an \emph{equivalence relation} on $X$, that is, it fulfills the reflexive, transitive and symmetric ($x \sim y$ if and only if $y \sim x$ $\forall x,y \in X$) properties. Notice, a preorder $\preceq$ is a partial order on the quotient set $X/\mathord{\sim} = \{[x]|x\in X\}$, consisting of all equivalence classes $[x] = \{y\in X| y\sim x\}$. In case $x \preceq y$ and $\neg(x\sim y)$
for some $x,y \in X$ we say $y$ is \emph{strictly preferred} to $x$, denoted by $x \prec y$. If $\neg(x \preceq y)$ and $\neg(y \preceq x)$, we say $x$ and $y$ are \emph{incomparable}, denoted by $x \bowtie y$. Whenever there are no incomparable elements a preordered space is called \emph{total}. By the Szilprajn extension theorem \cite{szpilrajn1930extension}, \textit{every partial order can be extended to a total order}, that is, to a partial order that is total. Notice, the set $\mathbb P_\Omega$ of probability distributions on $\Omega$ equipped with the uncertainty preorder $\preceq_U$ forms a non-antisymmetric preordered space, because equivalent elements are only equal up to permutations \cite{arnold2018majorization}.

A real-valued function $f:X \rightarrow \mathbb{R}$ is called a \emph{monotone} if $x \preceq y$ implies $f(x) \leq f(y)$. If also the converse is true, then $f$ is called a \emph{utility function}. Furthermore, if $f$ is a monotone and $x \prec y$ implies $f(x)<f(y)$, then $f$ is called a \emph{strict monotone} (or a \emph{Richter-Peleg function} \cite{alcantud2016richter}).

\begin{definition}[Injective monotones] 
A monotone $f:X\to \mathbb R$ on a preordered space $(X,\preceq)$ is called an \emph{injective monotone}  if $f(x)=f(y)$ implies $x \sim y$, that is, if $f$ is injective considered as a function on the quotient set $X/\mathord{\sim}$.
\end{definition}

Clearly, an injective monotone is also a strict monotone, since $x\prec y$ and $f(x)=f(y)$ contradicts injectivity. The converse is not true, for example, Shannon entropy is a strict monotone for the uncertainty preorder $\preceq_U$ (Appendix \ref{majo}) but not an injective monotone, nor a utility. In fact, preorders that have an injective monotone form a class in between preorders that have a strict monotone and preorders that have a utility function.  

\begin{proposition} ~
\label{monotones different}
\begin{enumerate}[label=(\roman*)]
\item There are preorders with strict monotones but without injective monotones.
\item There are preorders with injective monotones and without utility functions.
\end{enumerate}
\end{proposition}

\begin{proof}
$(i)$ Consider $(\mathcal{P}(\mathbb{R}), \preceq)$, the power set $\mathcal P(\mathbb R)$ of the reals equipped with the preorder $\preceq$ defined by $U\preceq V$ if and only if $U=V$, or $U=\{0\}$ and $V=\{1\}$. Then $v: \mathcal{P}(\mathbb{R}) \rightarrow \mathbb{R}$, given by $v(\{1\})=1$ and $v(U)=0$ $\forall U \neq \{1\}$, is a strict monotone. However, there cannot be injective monotones, because here $|\mathcal P(\mathbb R)/{\sim}| = |\mathcal P(\mathbb R)|$ and by Cantor's theorem the cardinality of $\mathbb R$ is strictly smaller than the cardinality of $\mathcal P(\mathbb R)$.

$(ii)$ Consider $(\mathbb{R}, \preceq)$, where $x \preceq y$ if and only if $x \leq y$ and $x,y \neq 0$ or $x=y=0$. The identity $I: \mathbb{R} \rightarrow \mathbb{R}$ is an injective monotone. However, $(\mathbb{R}, \preceq)$ is non-total since $0 \bowtie x$ $\forall x \in \mathbb{R}/ \{0\}$ and, thus, has no utility function.
\end{proof}

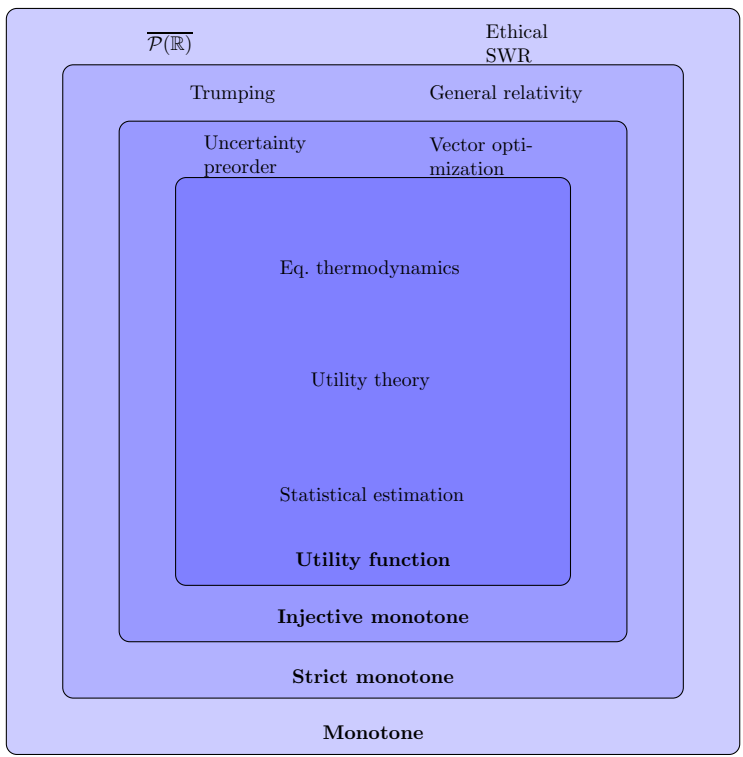
\begin{figure}[!t]
\centering
\begin{tikzpicture}[scale=0.75, every node/.style={transform shape}]
\node[rounded corners, draw, fill=blue!20 ,text height = 13cm,minimum     width=13cm, label={[anchor=south,above=1.5mm]270: \textbf{Monotone}}]  (main) {};
\node[rounded corners, draw, fill=blue!30, text height =11cm, minimum width = 11cm,label={[anchor=south,above=1.5mm]270: \textbf{Strict monotone}}] at (main.center)  (semi) {};
\node[rounded corners, draw,fill=blue!40, text height = 9cm, minimum width = 9cm,label={[anchor=south,above=1.5mm]270:\textbf{Injective monotone}}] at (main.center) (active) {};
\node[rounded corners, draw,fill=blue!50, text height = 7cm, minimum width = 7cm,label={[anchor=south,above=1.5mm]270:\textbf{Utility function}}] at (main.center) (non) {};
\node[text width=2cm] at (-3,6) {$\overline{\mathcal{P}(\mathbb{R})}$};
\node[text width=2cm] at (3,6) {Ethical SWR};
\node[text width=2cm] at (-2,4) {Uncertainty preorder};
\node[text width=2cm] at (2,4) {Vector optimization};
\node[text width=1.5cm] at (-2.5,5.1) {Trumping};
\node[text width=3cm] at (2.5,5.1) {General relativity};
\node[text width=2.2cm] at (0,0) {Utility theory};
\node[text width=3.3cm] at (0,2) {Eq.~thermodynamics};
\node[text width=3.3cm] at (0,-2) {Statistical estimation};
\end{tikzpicture}
\caption{Classification of preordered spaces according to the existence of various classes of monotones. The newly introduced class of preorders with injective monotones is strictly contained within the class of preorders with strict monotones, is strictly larger than the class of preorders with utility functions, and contains commonly used examples such as the uncertainty preorder and vector optimization. See Section \ref{conclu} and Appendix \ref{no strict mono} for a short description of the other examples in the figure. Notice we have classified both trumping and general relativity according to current knowledge, they may have injective monotones (see Section 
\ref{conclu}).}
\label{fig:classification}
\end{figure}

Since every preorder has a monotone (constant functions) and there are preorders without strict monotones (see Appendix \ref{no strict mono} for an example), we arrive at the picture shown in Figure \ref{fig:classification}. Notice, the closer we are to the center, the better a monotone represents the underlying preorder. In particular, injective monotones contain more information about the preorder than strict monotones. 

Nevertheless, in well-behaved cases, it is possible to construct an injective monotone out of a strict monotone. A negative example is the strict monotone that appears in the proof of $(i)$ in Proposition \ref{monotones different}, which  maps uncountably many incomparable elements to a single number (zero). If a strict monotone fails to be an injective monotone because of only countably many points, however, then it can easily be turned into an injective monotone by consecutive elimination.   

\begin{proposition}
\label{ R-P implies injective monotone}
A preordered space $(X, \preceq)$ has an injective monotone if and only if it has a strict monotone $f$ whose non-injective set
\begin{equation*}
I_f:=\{ x \in X| \text{ }\exists y \in X \text{ }s.t. \text{ } f(x)=f(y) \text{ and } x \bowtie y \}
\end{equation*}
is countable.
\end{proposition}

\begin{proof}
By definition, for an injective monotone $f$ we have $I_f= \emptyset$. Conversely, consider a strict monotone $f$ with a countable non-injective set. Given a numeration $\{x_n\}_{n \geq 0}$ of $I_f$, define $f_0: X \rightarrow \mathbb{R}$ by
\begin{equation*}
    f_0(x):= 
    \begin{cases}
    f(x) + 1 & \text{if} \text{ } f(x) \geq f(x_0) \text{ and } \neg(x \sim x_0)\\
    f(x) & \text{else.}
    \end{cases}
\end{equation*}
Notice, by definition, $\forall x,y \in X$, $f(x) \leq f(y)$ implies $f_0(x) \leq f_0(y)$, $I_{f_0} \subset I_f$, and $f_0$ is injective up to equivalence at $x_0$, in particular $x_0 \notin I_{f_0}$. Therefore, we can consecutively eliminate the elements in $I_f$ by defining for all $n\in\mathbb N$, $f_n(x) \coloneqq f_{n-1}(x) + 2^{-n}$ if $f_{n-1}(x) \geq f_{n-1}(x_n)$ and $\neg(x \sim x_n)$, and $f_n(x) \coloneqq f_{n-1}(x)$ otherwise, analogously to $f_0$. It is then straightforward to see that the pointwise limit $c(x)\coloneqq \lim_{n \to \infty} f_n(x)$ exists for all $x\in X$ and that $c$ is an injective monotone. 
\end{proof}

Notice, the technique in the proof of Proposition \ref{ R-P implies injective monotone} does not work if $I_f$ is uncountable.
In particular, it cannot be used to construct an injective monotone from Shannon entropy  $f=H$ for the uncertainty preorder $\preceq_U$, because if $N\coloneqq |\Omega|\geq 3$ then for all $c\in (0,\log N)$ there are $p,q \in \mathbb P_\Omega$ with $c=H(p)=H(q)$ but $p\bowtie q$ (see Appendix \ref{majo}). In other words, we can construct an injective map $g: ( 0,\log N ) \rightarrow I_H$ and thus $I_H$ has the same cardinality as $\mathbb{R}$, in particular $I_H$ is not countable.

\section{Relating monotones to optimization} \label{monotones and optimization}

An element $x\in X$ is called a \emph{maximal element of $\preceq$} if there exists no $y\in X$ such that $x\prec y$. For any $B\subseteq X$, an element $x \in B$ is called a \emph{maximal element of $\preceq$ in $B$} if there exists no $y\in B$ such that $x\prec y$. 

\begin{definition}[Representing maximal elements] \label{def:optima}
We say, a function $f:X \rightarrow \mathbb{R}$ is \emph{effective for $B \subseteq X$} if $\mathrm{argmax}_B \, f \neq \emptyset$, where $\mathrm{argmax}_B  f \coloneqq \{x \in B|\not \exists y \in B \text{ such that } f(x)<f(y) \}$. We say, a function $f: X \rightarrow \mathbb{R}$ \emph{represents maximal elements of $\preceq$}, if for any $B\subseteq X$
\[
\mathrm{argmax}_B  f \subseteq B^{\preceq}_M \, ,
\]
where $B^{\preceq}_M$ denotes the set of maximal elements of $\preceq$ in $B$. Similarly, we say, a function $f$ \emph{injectively represents maximal elements of $\preceq$}, if for any $B \subseteq X$ for which $f$ is effective, there exists $x_0 \in B^\preceq_M$ such that 
\[
\mathrm{argmax}_{B} f = [x_0]|_B
\] 
where $[x_0]|_B$ is the equivalence class of $x_0$ restricted to $B$. Moreover, we say, $(X, \preceq)$ has an \emph{(injective) optimization principle} if there exists a function $f:X \rightarrow \mathbb{R}$ which (injectively) represents maximal elements of $\preceq$.
\end{definition}

Even though Shannon entropy does not represent $\preceq_U$ as a utility, its property as a stict monotone guarantees that its maxima are in fact maximal elements  of $\preceq_U$, i.e. $H$ represents maximal elements  of the uncertainty preorder according to Definition \ref{def:optima}. Indeed, any $p \in \mathrm{argmax}_B H$ is a maximal element of $\preceq_U$ for any $B\subseteq \mathbb P_\Omega$ on which $H$ is effective, as $p\prec q$ for some $q\in B$ would lead to the contradiction $H(q)>H(p)$. In fact, representing maximal elements  is closely related to being a monotone for preorders in general.

\begin{proposition}
\label{optimization charac of R-P}
Given a preordered space $(X, \preceq)$ and a monotone $u:X \rightarrow \mathbb{R}$, then 
\begin{enumerate}[label=(\roman*)]
\item $u$ is a strict monotone if and only if $u$ represents maximal elements  of $\preceq$. 
\item $u$ is an injective monotone if and only if $u$ injectively represents maximal elements  of $\preceq$.
\end{enumerate}
\end{proposition}

\begin{proof}
$(i)$ If $u$ is a strict monotone, then $\mathrm{argmax}_B u \subseteq B_M^\preceq$ (by the same argument as for entropy). Conversely, consider $x,y\in X$ with $x\prec y$. For $B\coloneqq \{x,y\}$, we have $B_M^\preceq = \{y\}$ and thus $y = \mathrm{argmax}_B u$, i.e.~$u(x)<u(y)$.

$(ii)$ For any $B\subseteq X$ on which $u$ is effective, if $x,y$ $\in \text{argmax}_B u$, we have $u(x)=u(y)$ and, since $u$ is an injective monotone, $x\sim y$. Conversely, consider $x,y\in X$ and $B:=\{x,y\}$. 
If $u(x)=u(y)$ then by hypothesis $\{x, y\} = \text{argmax}_{x\in B} \{u(x)\}=[x_0]|_B$ for some $x_0 \in B$. In particular, $x \sim y$.
\end{proof}

Notice, for the ``if'' part in $(ii)$ we do not have to assume that $u$ is a monotone, that is, if the maxima of some real-valued function $u$ form an equivalence class in the set of  maximal elements, then it already follows that $u$ is a monotone.

For any preordered space $(X,\preceq)$, thus, the existence of a strict monotone implies the existence of an optimization principle and the existence of an injective monotone is equivalent to the existence of an injective optimization principle. One can contrast the \emph{global} injective representation of maximal elements  which characterizes injective monotones in Proposition \ref{optimization charac of R-P} with \emph{local} approaches, for some specific $B\subseteq X$, present in the literature  \cite{white1980notes,bevilacqua2018multiobjective}.

Choosing a particular strict monotone $u$ and optimizing it in a set $B$ might, however, not yield all the maximal elements  in $B_M^\preceq$. {\color{black}For example, take $p,q\in\mathbb P_\Omega$ with $p\bowtie q$ and $H(p) < H(q)$, then $B=\{p,q\}$ has the two maximal elements $p$ and $q$, but $\mathrm{argmax}_B\, H = \{p\}$. Notice, this is not only an issue for trivial examples like this, but also happens for the maximum entropy principle with linear constraint sets. In particular, if $B=\{ p \,| \, \langle E \rangle = c\}$, for a given random variable $E$ and some $c\in\mathbb R$, crosses two incomparable elements that turn out to be maximal (see Figure \ref{fig:maxent}), then only part of the actual maximal elements of $\preceq_U$ can be found by maximizing entropy.}

Similarly, while optimizing an \emph{injective} monotone in a set $B$ results in equivalent elements, in general we only find a slice of the set of all maximal elements in $B$. In fact, for every maximal element $x$ in $B_M^\preceq$ we can construct an injective monotone $c$ such that $x\in \mathrm{argmax}\,c$ (e.g., in the proof of Proposition \ref{strict mono equi r-p repre} below, take $c_x$ if $x\in A_c$ and $c$ otherwise). This means that the problem of selecting a maximal equivalence class can be replaced by the problem of selecting an injective monotone.

In the following section, we show that injective monotones exist for a large class of preorders, including the uncertainty preorder.

\begin{figure}[ht]
\centering
\begin{tikzpicture}[scale=0.6]

\draw (0,6)--(-6,-6)--(6,-6)--cycle;
\draw (-3,0)--(3,0)--(0,-6)--cycle;
\draw (0,0)--(-1.5,-3)--(1.5,-3)--cycle;
\draw [fill=blue!30] (-3,0)--(0,6)--(3,0)-- node[above=1cm] {\textbf{$\succ_U$}} cycle;
\draw [fill=blue!30] (3,0)--(6,-6)-- node[above=1cm] {\textbf{$\succ_U$}} (0,-6)--cycle;
\draw [fill=blue!30] (-6,-6)-- node[above=1cm] {\textbf{$\succ_U$}} (0,-6)--(-3,0)--cycle;
\draw [fill=blue!40] (0,0)--(-1.5,-3)-- node[above=0.5cm] {\textbf{$\prec_U$}}(1.5,-3)--cycle;
\draw [fill=white] (0,0)--(-3,0)--(-1.5,-3)--cycle;
\draw [fill=white] (0,0)--(3,0)--(1.5,-3)--cycle;
\draw [fill=white] (0,-6)--(-1.5,-3)--(1.5,-3)--cycle;

\draw [line width=0.4mm] (-1.5,-6)-- (-1.5,3);
\draw [line width=0.0mm] (-1.5,1)-- node[ below left=0.05cm] {\textbf{$B$}} (-1.5,1);
\filldraw (-1.5,-3) circle (3pt) node[left=0.3mm] {$p$};
\filldraw (-1.5,-2) circle (3pt) node[left=0.3mm] {$q$};

\end{tikzpicture}
\caption{Example for when the maximum entropy principle does not yield all maximal elements of $\preceq_U$ in some $B \subseteq P_\Omega$. Here, we show the usual visualization of the 2-simplex, that is, the set of all probability distributions in $P_\Omega$ for $|\Omega|=3$. Let the energy function $E$ be given by $E(x_1)\coloneqq 1$, $E(x_2)\coloneqq-1$, and $E(x_3)\coloneqq 0$, and let $B$ be given by the constraint $\langle E \rangle = \frac{1}{4}$, represented by the vertical line. The distribution $p=(1/2, 1/4, 1/4)$ is a maximal element in $B$, because any other element of $B$ is either smaller than $p$ (belongs to an outer blue region) or incomparable (belongs to the white region). However, $q=(9/20,,4/20,7/20)$ is in $B$ and $H(p) < H(q)$. As a result, $p$ is a maximal element of $B$ which is not obtained via the maximum entropy principle.}
\label{fig:maxent}
\end{figure}
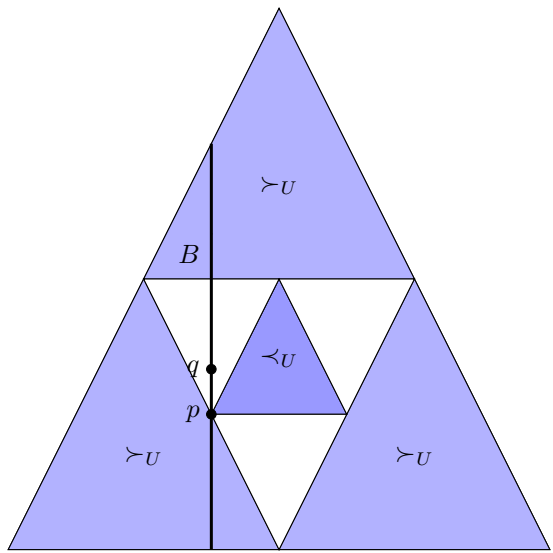

\section{Relating monotones to multi-utilities} \label{monot and multi-ut}

Although it is not possible to capture all information about a non-total preorder using a single
real-valued function, a family of functions may be used instead.
A family $V$ of real-valued functions $v: X \rightarrow \mathbb{R}$ is called a \emph{multi-utility (representation) of $\preceq$} if 
\begin{equation*}
x \preceq y \iff v(x) \leq v(y)  \text{ }\forall v \in V \, .
\end{equation*}
Whenever a multi-utility consists of strict monotones it is called a \emph{strict monotone} (or \emph{Richter-Peleg} \cite{alcantud2016richter}) \emph{multi-utility (representation) of $\preceq$}. Analogously, if the multi-utility consists of injective monotones, we call it an \emph{injective monotone multi-utility (representation) of $\preceq$}. 

It is straightforward to see that every preordered space $(X,\preceq)$ has the multi-utility $(\chi_{i(x)})_{x\in X}$, where $\chi_A$ denotes the characteristic function of a set $A$ and $i(x)\coloneqq \{y\in X| x\preceq y\}$ \cite{ok2002utility}. Moreover, if there exists a strict monotone $u$, then a multi-utility $U$ only consisting of strict monotones can easily be constructed from a given multi-utility $V$ by $U\coloneqq \{v + \alpha u\}_{v \in V,\alpha>0}$ \cite{alcantud2013representations}. Even though this construction does not work directly in the case of injective monotones, a simple modification does, where special care is given to incomparable elements.

\begin{proposition}
\label{strict mono equi r-p repre}
Let $(X,\preceq)$ be a preordered space. There exists an injective monotone if and only if there exists an injective monotone multi-utility. 
\end{proposition}

\begin{proof}
Consider w.l.o.g. an injective monotone $c:X \rightarrow (0,1)$ and
\begin{equation}
\label{Ac}
    A_c:=\{ x \in X| \exists y \in X \text{ s.t. } x \bowtie y,\text{ } c(x) < c(y)\}, 
\end{equation}
i.e.~the part of $X$ that has incomparable elements $y$ with strictly larger values of $c$. For all $x \in A_c$, let $c_x \coloneqq c + \chi_{i(x)}$. Notice, by construction $c_{x}(y) = c(y) < 1 \leq c_{x}(x)$ for all $y\in X$ with $x\bowtie y$. By using $c(X)\subseteq (0,1)$ and the fact that $c$ is an injective monotone, it is straightforward to see that $C\coloneqq \{c\}\cup \{c_x\}_{x\in A_c}$ is an injective monotone multi-utility.
\end{proof}

Note that the injective monotone multi-utility in the proof of Proposition \ref{strict mono equi r-p repre} can be chosen to have cardinality of at most $\mathfrak{c}$, the cardinality of the continuum, because it is enough to have one $c_x$ per equivalence class $[x]\in X/\mathord{\sim}$, and, whenever an injective monotone exists, $|X/\mathord{\sim}| \leq \mathfrak{c}$.

The cardinality of multi-utilities plays an important role. In particular, special interest lies in preordered spaces with \emph{countable} multi-utilities. In practice, countable multi-utilities are often used to define preordered spaces. {\color{black}For example, the uncertainty preorder $\preceq_U$ is defined in \eqref{uncert rela} by a countable (finite) multi-utility.} Also, many applications in multicriteria optimization \cite{bevilacqua2018multiobjective,ehrgott2005multicriteria} rely on preordered spaces defined by countable multi-utilities. It turns out that for the existence of strict monotones, {\color{black}such as entropy for $\preceq_U$}, it is sufficient to have a countable multi-utility \cite[Section 4]{alcantud2016richter}. Here, we show that countable multi-utilities actually imply the existence of \emph{injective} monotones, which, due to Proposition \ref{monotones different}, improves upon \cite{alcantud2016richter}.

\begin{proposition}
\label{countable implies injective monotone}
If, for a given preordered space $(X,\preceq)$, there exists a countable multi-utility, then there exists an injective monotone.
\end{proposition}

This means that the class of preordered spaces where countable multi-utilities exist is contained in the class of preordered spaces where an injective monotone exists (cf. Figure \ref{fig:classification}). However, there exist preordered spaces with injective monotones, i.e., by Proposition \ref{strict mono equi r-p repre}, with injective monotone multi-utilities of cardinality $\mathfrak{c}$, but without countable multi-utilities (see Proposition \ref{inj mono no count mu}). 

{\color{black}
For the uncertainty preorder $\preceq_U$, which is defined in \eqref{uncert rela} through a finite multi-utility,  Proposition \ref{countable implies injective monotone} therefore guarantees the existence of injective monotones. Moreover, we can see a possible construction in \eqref{eq:constructInjectiveMonotone} below. }

By a slight adaptation of the proof of Proposition \ref{countable implies injective monotone}, we obtain the stronger

\begin{proposition}
\label{countable implies injective monotone multi}
For a given preordered space $(X,\preceq)$, there exists a countable multi-utility if and only if there exists a countable multi-utility only consisting of injective monotones.
\end{proposition}

This improves upon \cite[Proposition 4.1]{alcantud2016richter}, where it is shown that a countable multi-utility exists if and only if a countable strict monotone multi-utility exists. Notice, however, while for the proof in \cite{alcantud2016richter}, one can simply modify each member of a given multi-utility separately---similarly as we did for the construction in Proposition \ref{strict mono equi r-p repre}---our proof of Proposition \ref{countable implies injective monotone multi} relies on a more indirect technique, where each member of the resulting injective monotone multi-utility does not have a direct relationship to a non-injective member of the given multi-utility.

For the proofs of Propositions \ref{countable implies injective monotone} and \ref{countable implies injective monotone multi} we rely on the following basic facts, the proofs of which can be found in the appendix. 

\begin{lemma} \label{series} Let $X$ be a set. Given $r\in (0, \frac{1}{2})$ and a countable family $(A_n)_{n\geq 0}$ of subsets $A_n \subseteq X$, define the function $c: X\to \mathbb R$ by
\begin{equation} \label{eq:constructInjectiveMonotone}
c(x):= \sum_{n\geq 0} r^{n} \chi_{A_n}(x) \, .
\end{equation}
Then, $c(x)<c(y)$ if and only if, for the first $m\in\mathbb N$ with $\chi_{A_m}(x) \not = \chi_{A_m}(y)$, we have $\chi_{A_m}(x) < \chi_{A_m}(y)$.
\end{lemma}

The following characterizations of injective monotones and countable multi-utilities follow by straightforward manipulations of their definitions.

\begin{lemma} \label{basic characterizations} Let $(X,\preceq)$ be a preordered space. A monotone $u$ is an injective monotone if and only if 
\begin{equation}
x\prec y \ \Rightarrow \ u(x)<u(y) \ \text{ and } \ x\bowtie y \ \Rightarrow \ u(x)\not = u(y) \, .
\end{equation} 
A collection $U$ of monotones is a multi-utility if and only if 
\begin{equation}\label{charac:multi-util}
\neg (y\preceq x) \ \Rightarrow  \ \exists u\in U \text{ s.t. } u(x)< u(y) \, .
\end{equation}
\end{lemma}

A subset $A\subseteq X$ of a preordered space $(X,\preceq)$ is called \emph{decreasing} if for all $x\in A$, $y\preceq x$ implies $y\in A$. Analogously, a subset $A\subseteq X$ is called \emph{increasing}, if for all $x\in A$, $x\preceq y$ implies that $y\in A$ \cite{mehta1986existence}. We say a family $(A_n)_{n\in\mathbb N}$ of subsets $A_n\subseteq X$ \emph{separates $x$ from $y$}, if there exists $n\in \mathbb N$ with $x\not\in A_n$ and $y\in A_n$. 

\begin{lemma}\label{lemma:sep}
Let $(A_n)_{n\geq 0}$ be a family of increasing sets.
\begin{enumerate}[label=(\roman*)]
\item If, for all $x,y\in X$ with $x \prec y$, $(A_n)_{n\geq 0}$ separates $x$ from $y$, then the function
$c: X\to \mathbb R$ defined in \eqref{eq:constructInjectiveMonotone} is a strict monotone for all $r\in (0,1)$.
\item If in addition, for all $x,y\in X$ with $x \bowtie y$, $(A_n)_{n\geq 0}$ separates $x$ from $y$, or
$y$ from x, then $c$ is an injective monotone for all $r\in (0,\frac{1}{2})$.
\end{enumerate}
\end{lemma}

Notice, the construction of strict and injective monotones in Lemma \ref{lemma:sep} is based on Lemma \ref{series} and is analogous to constructions that appear in the literature, where one typically uses a value of $r=\frac{1}{2}$ (e.g.~\cite{mehta1977topological,ok2002utility,alcantud2016richter}). The requirement of $r<\frac{1}{2}$ in Lemma \ref{series} and \ref{lemma:sep} ensures that the resulting monotone is injective. In fact, as can be seen from the proof of Lemma \ref{series} in the appendix, for $r \in (0,1)$ we have $r^m = \frac{r}{1-r} \sum_{n>m} r^n$. A value of $r\in (0,\frac{1}{2})$ thus enables the strict estimate $r^m > \sum_{n>m} r^n$, which is exactly where the injectivity up to equivalence of $c$ in Lemma \ref{lemma:sep} rests.

\begin{proof}[Proof of Proposition \ref{countable implies injective monotone}]
For a countable multi-utility $(u_m)_{m\in M}$ and $q\in\mathbb Q$, consider the increasing sets $A_{m,q} \coloneqq u_m^{-1}([q,\infty))$. It suffices to show that $(A_n)_{n\geq 0}$, where $A_n \coloneqq A_{m_n,q_n}$ for some enumeration $n\mapsto (m_n,q_n)$ of $M\times \mathbb Q$, satisfies $(i)$ and $(ii)$ in Lemma \ref{lemma:sep}. If $x\prec y$ or $x\bowtie y$, then, by \eqref{charac:multi-util}, in both cases there exists $m\in M$ with $u_m(x)<u_m(y)$. Hence, we can choose $q\in \mathbb Q$ with $u_m(x)<q<u_m(y)$, in particular, $x\not\in A_{m,q}$ and $y\in A_{m,q}$.
\end{proof}

\begin{proof} [Proof of Proposition \ref{countable implies injective monotone multi}]
Let $(u_m)_{m\in M}$ be a countable multi-utility and let $c$ be an injective monotone of the form \eqref{eq:constructInjectiveMonotone} constructed from the increasing sets $A_n$ in the proof of Proposition \ref{countable implies injective monotone}. We define, for any pair $(m,p) \in \mathbb{N}$ such that $m<p$, $\varphi_{m,p}: \mathbb{N} \to \mathbb{N}$ which permutes $m$ and $p$ without changing any other natural number. For each $\varphi_{m,p}$, we define an injective monotone $c_{m,p}$ of the form \eqref{eq:constructInjectiveMonotone} constructed from $(A_{\varphi_{m,p}(n)})_{n\geq 0}$, the increasing sets used to define $c$ reordered by $\varphi_{m,p}$. Since $\{c\}\cup\{c_{m,p}\}_{(m,p) \in \mathbb{N}^2, \text{ } m<p}$ is composed of injective monotones, it suffices to show \eqref{charac:multi-util} holds to conclude there exists a countable multi-utility composed of injective monotones. Consider, thus, $x,y \in X$ such that $\neg (y\preceq x)$. If $x \prec y$, then $c(x)<c(y)$ by definition. Assume now $x \bowtie y$. If $c(x)<c(y)$, then we have finished. Otherwise, we have $x \in A_m$ and $y \not \in A_m$ for the first $m \in \mathbb{N}$ such that $\chi_{A_m}(x) \not = \chi_{A_m}(y)$ by Lemma \ref{series}. Since there exists some $p \in \mathbb{N}$ $p>m$ such that $y \in A_p$ and $x \not \in A_p$, the first $n \in \mathbb{N}$ such that $\chi_{A_{\varphi_{m,p}}(n)}(x) \not = \chi_{A_{\varphi_{m,p}(n)}}(y)$ is $n=m$. We conclude $c_{m,p}(x)<c_{m,p}(y)$ by Lemma \ref{series}, since we have $\chi_{A_{\varphi_{m,p}}(m)}(x)  = \chi_{A_p}(x) < \chi_{A_p}(y) = \chi_{A_{\varphi_{m,p}(m)}}(y)$.
\end{proof}

Countable separating families such as the ones in Lemma \ref{lemma:sep} have been used to characterize preordered spaces with continuous utility functions \cite{herden1989existence}, generalizing theorems of Peleg and Mehta \cite{mehta1981recent}. In a similar spirit, Alcantud et al.~\cite{alcantud2013representations} extend a result by Bosi and Zuanon \cite{bosi2013existence} about upper semicontinuous multi-utilities based on separating families, showing that \emph{there exists a countable multi-utility if and only if there exists a family of decreasing subsets that $\forall x,y\in X$ with $\neg (y \preceq x)$ separates $x$ from $y$} \cite[Proposition 2.13]{alcantud2013representations}. Using Lemma \ref{lemma:sep} and the characterizations in Lemma \ref{basic characterizations}, we immediately get the following analogous characterizations for preorders with strict and injective monotones, the proofs of which can be found in Appendix \ref{proof prop 7}.

\begin{proposition}
\label{set charac monotones}
Let $(X,\preceq)$ be a preordered space.
\begin{enumerate}[label=(\roman*)]
 \item There exists a strict monotone if and only if there exists a family of increasing subsets that $\forall x,y\in X$ with $x \prec y$ separates $x$ from $y$.
 \item There exists an injective monotone if and only if there exists a family of increasing subsets that satisfies (i) and (ii) in Lemma \ref{lemma:sep}.
\end{enumerate}
\end{proposition}

 Countable separating families are a useful tool to improve the classification of preordered spaces by monotones. In particular, we use them in Propostion \ref{inj mono no count mu} to show the converse of Proposition \ref{countable implies injective monotone} is false, that is, there are preordered spaces where injective monotones exist and countable multi-utilities do not.

\begin{proposition}
\label{inj mono no count mu}
There are preordered spaces with injective monotones and without countable multi-utilities.
\end{proposition}

\begin{proof}
Consider $X \coloneqq [0,1] \cup [2,3]$ equipped with $\preceq$ where 
\begin{equation}
\label{order def 1}
        x \preceq y \iff 
    \begin{cases}
    x=y\\
    x  \in [0,1],\text{ } y \in [2,3] \text{ and } y \neq x+2
    \end{cases}
    \end{equation}
    $\forall x,y \in X$ (see Figure \ref{fig: inj mono and no count mu} for a representation of $\preceq$). Notice $(X,\preceq)$ is a preordered space and the identity map $i_d: X \to \mathbb{R}$ is an injective monotone. We will show any family $(A_i)_{i \in I}$, where $A_i \subseteq X$ is increasing $\forall i \in I$ and $\forall x,y \in X$ such that $\neg (y \preceq x)$ there exists some $i \in I$ such that $x \not \in A_i$ and $y \in A_i$, is uncountable. Since the existence of some $(A_i)_{i \in I}$ with those properties and countable $I$ is equivalent to the existence of a countable multi-utility \cite[Proposition 2.13]{alcantud2013representations}, we will get there is no countable multi-utility for $X$. Consider a family $(A_i)_{i \in I}$ with the properties above and, for each $x \in [0,1]$, $y_x \coloneqq x+2$. Since $x \bowtie y_x$ by definition, there exists some $A_x \in (A_i)_{i \in I}$ such that $x \in A_x$ and $y_x \not \in A_x$. We fix such an $A_x$ for each $x \in [0,1]$ and consider the map $f:[0,1] \to (A_i)_{i \in I}$, $x \mapsto A_x$. Given $x,z \in [0,1]$ $x \neq z$, if we assume $z \in A_x$, then, since $A_x$ is increasing and $z \prec y_x$ as $y_x \neq z+2$, we would have $y_x \in A_x$, a contradiction. Notice, analogously, we get a contradiction if we assume $x \in A_z$  and, therefore, $A_x \not = A_z$. Thus, $A_x=A_z$ implies $x=z$ and we have, by injectivity of $f$, $|[0,1]| \leq |(A_i)_{i \in I}|$. As a consequence, $X$ has no countable multi-utility.
\end{proof}

\begin{figure}[!tb]
\centering
\begin{tikzpicture}
\node[rounded corners, draw,fill=blue!20, text height = 2.5cm, minimum width = 11cm,xshift=4cm,label={[anchor=west,left=.1cm]180:\textbf{B}}] {};
\node[rounded corners, draw,fill=blue!20, text height = 2.5cm, minimum width = 11cm,xshift=4cm,yshift=-4cm,label={[anchor=west,left=.1cm]180:\textbf{A}}] {};
    \node[main node] (1) {$x+2$};
    \node[main node] (2) [right = 2cm  of 1]  {$y+2$};
    \node[main node] (3) [right = 2cm  of 2]  {$z+2$};
    \node[main node] (4) [below = 2cm  of 1] {$x$};
    \node[main node] (5) [right = 2cm  of 4] {$y$};
    \node[main node] (6) [right = 2cm  of 5] {$z$};

    \path[draw,thick,->]
    (4) edge node {} (2)
    (4) edge node {} (3)
    (5) edge node {} (1)
    (5) edge node {} (3)
    (6) edge node {} (1)
    (6) edge node {} (2)
    ;
\end{tikzpicture}
\caption{Representation of a preordered space, defined in Proposition \ref{inj mono no count mu}, where injective monotones exist and countable multi-utilities do not. In particular, we show $A \coloneqq [0,1]$, $B \coloneqq [2,3]$ and how $x,y,z \in A$, $x<y<z$, are related to $x+2,y+2,z+2 \in B$. Notice an arrow from an element $w$ to an element $t$ represents $w \prec t$.}
\label{fig: inj mono and no count mu}
\end{figure}
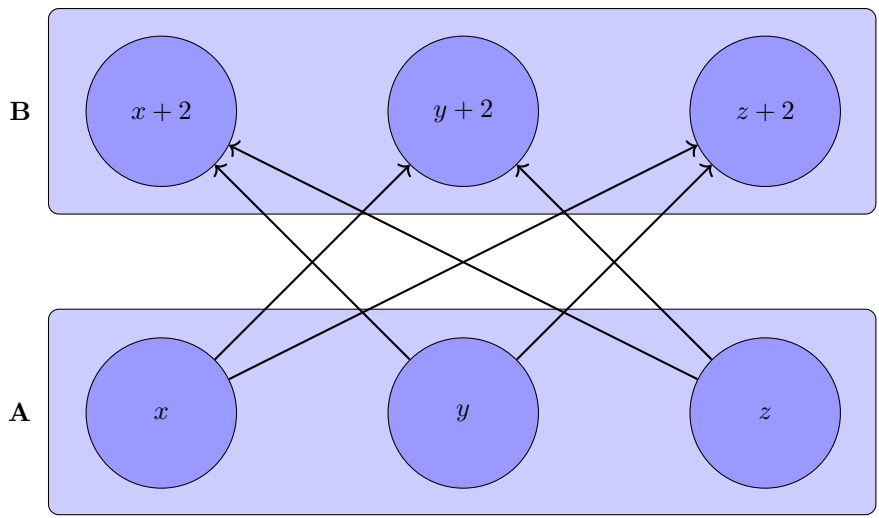

As we have seen in this section, the concept of separating families is closely related to the existence of monotones. In particular, this link is apparent when considering sets of the form $u^{-1}([q,\infty))$ for some monotone $u$ and $q\in \mathbb Q$, allowing to translate the two concepts into each other (see the proofs of Propositions \ref{countable implies injective monotone} and \ref{countable implies injective monotone multi}). There is another rich class of separability properties of preordered spaces providing necessary conditions for the existence of monotones, which could collectively be described by the term \emph{order separability}. Many important results from mathematical economics fall into this category, such as the Debreu Open Gap Lemma \cite{debreu1964continuity}, the Nachbin Separation Theorem \cite{nachbin1965topology}, Szpilrajn's theorem \cite{szpilrajn1930extension}, and Fishburn's theorem \cite[Theorem 3.1]{fishburn1970utility}. We discuss the role of injective monotones relative to order separability in the following section.

\section{Relating monotones to order separability} \label{monot and denseness}

\begin{table}[!t]
\caption{Separability properties of preordered spaces $(X,\preceq)$.}
\begin{center}
\begin{tabular}{ p{3cm} p{1.3cm} p{6.5cm} }
 \hline\noalign{\smallskip}
 Name& \hfil Object & \hfil Definition\\
 \hline\noalign{\smallskip}
order dense  & \hfil $Z \subseteq X$    & $\forall x,y \in X$ $x \prec y$ $\implies$ $\exists z \in Z$: $x \prec z \prec y$ \\
 Debreu dense&   \hfil $Z \subseteq X$  & $\forall x,y \in X$ $x \prec y$ $\implies$ $\exists z \in Z$: $x \preceq z \preceq y$  \\
 upper dense&   \hfil$Z \subseteq X$  & $\forall x,y \in X$ $x \bowtie y$ $\implies$ $\exists z \in Z$: $x \bowtie z \prec y$\\
 Debreu upper dense& \hfil$Z \subseteq X$  & $\forall x,y \in X$ $x \bowtie y$ $\implies$ $\exists z \in Z$: $x \bowtie z \preceq y$ \\
 order separable& \hfil $X$ & $\exists Z \subseteq X$ countable: $Z$ is order dense \\
 Debreu separable    & \hfil$X$ & $\exists Z \subseteq X$ countable: $Z$ is Debreu dense\\
 upper separable& \hfil$X$  & $\exists Z \subseteq X$ countable: $Z$ is order dense and upper dense\\
 Debreu upper separable& \hfil $X$  & $\exists Z \subseteq X$ countable: $Z$ is Debreu dense and Debreu upper dense \\
 \hline\noalign{\smallskip}
\end{tabular}
\end{center}
\end{table}

A subset $Z\subseteq X$, such that $x\prec y$ implies that there exists $z\in Z$ with $x\prec z \prec y$ is called \emph{order dense} \cite{ok2002utility,bridges2013representations}, and $Z$ is called \emph{order dense in the sense of Debreu} (or \emph{Debreu dense} for short) if $x\preceq z \preceq y$. Accordingly, we say that $(X, \preceq)$ is \emph{order separable} if there exists a countable order dense set \cite{mehta1986existence}, and \emph{Debreu separable} if there exists a countable Debreu dense set in $(X,\preceq)$. Notice, our definition of order separability is also known as \emph{weak separability} \cite{ok2002utility}.

It is well-known that a total preorder $\preceq$ has a utility function if and only if it is Debreu separable (e.g. \cite[Theorem 1.4.8]{bridges2013representations}). Moreover, if $\preceq$ is non-total, then Debreu separability still implies the existence of strict monotones \cite{herden2012utility,bridges2013representations,debreu1954representation}. The converse, however, is not true, i.e.~there are preordered spaces with strict monotones that are not Debreu separable. For example, any Debreu dense subset of $\smash{(\mathbb P_\Omega, \preceq_U)}$ is uncountable (if $|\Omega| > 2$)---see Appendix \ref{majo} for a proof. While Debreu separability is concerned with elements satisfying $x \prec y$, an analogous condition that is sufficient for the existence of injective monotones must also consider incomparable elements. 

We call a subset $Z\subseteq X$ \emph{upper dense} if $x \bowtie y$ implies that there exists a $z \in Z$ such that $x \bowtie z \prec y$ \footnote{Notice, for a fixed pair $x,y \in X$ where $x \bowtie y$ holds, there exist $z_1,z_2 \in Z$ such that $x \bowtie z_1 \prec y$ and $y \bowtie z_2 \prec x$. The same applies to upper density in the sense of Debreu, substituting $\prec$ by $\preceq$.}, and it is called \emph{upper dense in the sense of Debreu} (or \emph{Debreu upper dense} for short) if $x \bowtie z \preceq y$. Accordingly, $(X,\preceq)$ is called \emph{upper separable} if there exists a countable subset of $X$ which is both order dense and upper dense \cite{ok2002utility}, and $(X,\preceq)$ is called \emph{Debreu upper separable} if there exists a countable subset which is both Debreu dense and Debreu upper dense. We list all mentioned order denseness and separability properties in Table 1.

\begin{proposition}
\label{weakly upper implies countable mu}
If $(X,\preceq)$ is a Debreu upper separable preordered space, then there exists a countable multi-utility; in particular, there exists an injective monotone.
\end{proposition}

\begin{proof}
Consider a countable set $D$ given by Debreu upper separability. We will show 
\begin{equation*}
    x\preceq y \ \iff \ 
    \begin{cases}
    \chi_{i(d)}(x) \leq \chi_{i(d)}(y) \\
    \chi_{r(d)}(x) \leq \chi_{r(d)}(y) 
    \end{cases}
    \forall d \in D \, ,
\end{equation*}
where $i(d) \coloneqq \{ y\in X| d\preceq y \}$ and $r(d) \coloneqq \{ y\in X| d\prec y \}$. By transitivity $x \preceq y$ implies $\chi_{i(d)}(x) \leq \chi_{i(d)}(y)$ and $\chi_{r(d)}(x) \leq \chi_{r(d)}(y)$ $\forall d \in D$. If $\neg(x \preceq y)$ then either $y \prec x$ or $y \bowtie x$. If $y \prec x$ then there exists some $d \in D$ such that either $\chi_{i(d)}(x) > \chi_{i(d)}(y)$ or $\chi_{r(d)}(x) > \chi_{r(d)}(y)$. If $y \bowtie x$ then there exists some $d \in D$ such that $y \bowtie d \preceq x$ which means $\chi_{i(d)}(x) > \chi_{i(d)}(y)$. Since there exists a countable multi-utility, as we just showed, there is an injective monotone by Proposition \ref{countable implies injective monotone}.
\end{proof}

Since Debreu upper separability still requires a countable Debreu dense set, the converse of Proposition \ref{weakly upper implies countable mu} is again false due to the uncertainty preorder not being Debreu separable (Appendix \ref{majo}). 
However, as can be seen from the proof, if we remove Debreu denseness as a requirement, i.e. if we only require $D$ to be Debreu upper dense, then the only part of the proof that does not work is to follow from $y\prec x$ that there exists an element $v$ of the multi-utility with $v(x)>v(y)$. Since a strict monotone has exactly this property, we obtain the following proposition.

\begin{proposition}
\label{exists countable weakly upper dense implies equivalence R-P and injective monotone}
Consider $(X, \preceq)$ a preordered space. If there exists a countable Debreu upper dense set then the following are equivalent:
\begin{enumerate}[label=(\roman*)]
\item There exists a strict monotone.
\item There exists an injective monotone.
\item There exists a countable multi-utility.
\end{enumerate}
\end{proposition}

\begin{proof} Assume there exists a countable Debreu upper dense set $D\subseteq X$. It is enough to show that $(i)$ implies $(iii)$, which follows along the same lines as the proof of Proposition \ref{weakly upper implies countable mu}, but with the multi-utility consisting of $\{u\}\cup\{\chi_{i(d)}\}_{d\in D}$, where $u$ is a strict monotone.
\end{proof}

The situation in Proposition \ref{exists countable weakly upper dense implies equivalence R-P and injective monotone} corresponds exactly to the situation of the uncertainty preorder, which has a countable Debreu upper dense set (Appendix \ref{majo}) and, e.g., Shannon entropy as a strict monotone.

\section{Discussion} 
\label{conclu}

In this paper, we are mainly concerned with the introduction of injective monotones, their relation to other monotones, optimization, multi-utilities and order separability, and the application to the uncertainty preorder. The key contributions of our work are the following. First, 
we refine the classification of preordered spaces based on the existence of monotones. In particular, by extending known results for strict monotones to injective monotones, we find conditions for their existence from different perspectives: other classes of monotones, optimization principles, separating families of increasing sets, and (in particular, countable) multi-utilities. An overview of our conditions in relation to previous work can be found in Figure \ref{order vs functions}. Second, we introduce the notion of upper Debreu separability, an order separability property that allows to extend well-known results about the existence of monotones on Debreu separable spaces to countable multi-utilities and injective monotones. Finally, we apply our general results to the uncertainty preorder, defined on the space of probability distributions over finite sets, in particular, by establishing order separability properties.

\begin{figure}[t]
\centering
\begin{tikzpicture}[scale=0.7, every node/.style={transform shape}]
\node[rounded corners, draw,text height = 15cm,minimum     width=15cm, label={[anchor=north,below=1.5mm]90: \textbf{ all preorders $\cong$ $\exists$ monotone $\cong$ $\exists$ multi-utility}}]  (main) {};
\node[rounded corners, draw, fill=blue!10, text height =13.5cm, minimum width = 13.5cm,label={[anchor=north,below=1.5mm]90: \textbf{$\exists$ strict monotone $\cong$ $\exists$ strict monotone multi-utility}}] at (main.center)  (semi) {};
\node[rounded corners, text height =13.5cm, minimum width = 13.5cm,label={[anchor=north,below=5.5mm]112: \textbf{$\exists f$: argmax $f$ $\subseteq$ $X^{\preceq}_M$}}] at (main.center)  (semi) {};
\node (active) [rounded corners, draw, red,fill=blue!18, text height = 6.5cm, minimum width = 11.5cm,label={[anchor=north,below=1.5mm,red]90:\textbf{$\exists$ injective monotone $\cong$ $\exists$ injective monotone multi-utility}}] at (main.center) {};
\node [rounded corners, text height = 6.5cm, minimum width = 12.5cm,label={[anchor=north,below=5.5mm,red]125:\textbf{$\exists f$, $\exists x_0\in X_M^\preceq$: argmax $f$ = $\big[ x_0 \big]$}}] at (main.center) {};
\node[rounded corners, yshift=-0.5cm, draw,fill=blue!25, text height = 5cm, minimum width = 7.7cm,label={[anchor=south,above=1.5mm]270:\textbf{$\exists$ countable multi-utility}}] at (main.center) (non) {};
\node[rounded corners, draw, yshift=-0.5cm, fill=white, text height = 1cm, minimum width = 4.5cm,label={[anchor=north,below=1.5mm]90:\textbf{$\exists$ utility}}] at (main.center) (non) {};
\node[rounded corners, yshift=-0.5cm, text height = 1cm, minimum width = 1.5cm,label={[anchor=south,above=1.5mm]270:\textbf{$\cong$ total and D. separable}}] at (main.center) (non) {};
\node[ellipse, draw, yshift=0cm, xshift=-1.65cm, minimum height = 11.2cm, minimum width= 10.2cm,label={[anchor=north,above=5mm, black]270:\textbf{Debreu separable}} ] at (main.center){};
\node[ellipse, draw, red, yshift=0cm, xshift=-0.5cm, minimum height = 4cm, minimum width= 6.5cm, label={[anchor=north,below=7mm,red]90:\textbf{Debreu}} ] at (main.center){};
\node[ellipse, yshift=0cm, xshift=-0.5cm, minimum height = 4cm, minimum width= 6.5cm, label={[anchor=north,below=11mm,red]90:\textbf{upper separable}} ] at (main.center){};
\end{tikzpicture}
\caption{Classification of preordered spaces $(X,\preceq)$ in terms of representations by real-valued functions (boxes) and order properties (ellipses). We include known relations in black and our contributions in red. Notice, by Proposition \ref{exists countable weakly upper dense implies equivalence R-P and injective monotone}, the blue area is empty whenever there exists a countable Debreu upper dense set in $X$.}
\label{order vs functions}
\end{figure}
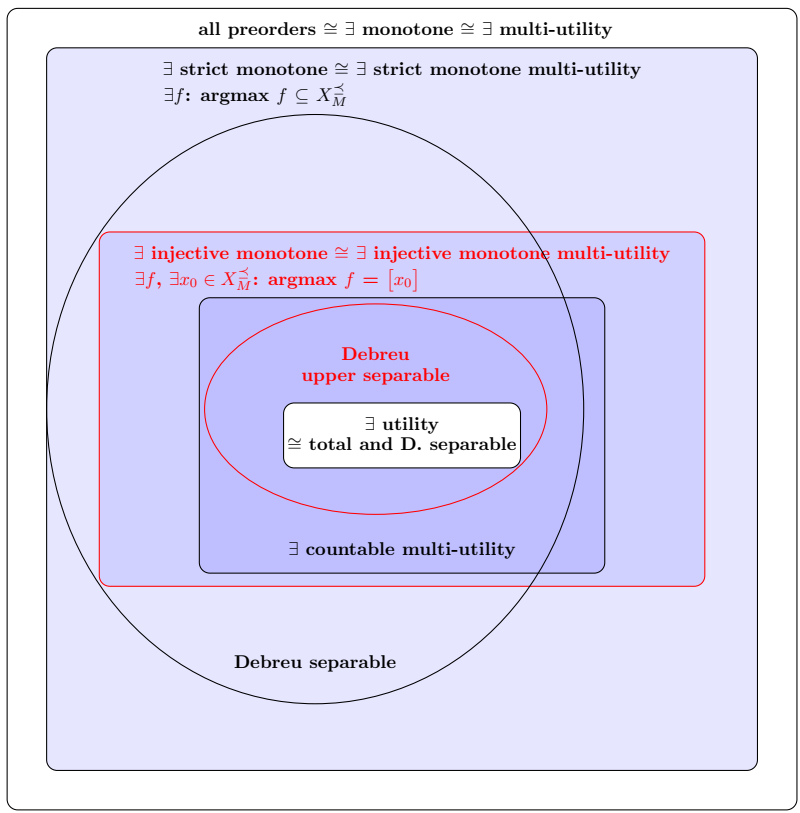

\paragraph{Hierarchy of preordered spaces.} A number of scientific disciplines rely on preorder spaces and their representation by monotones, as was already pointed out in \cite{campion2018survey,candeal2001utility,minguzzi2010time}.  In Figure \ref{fig:classification}, we classify the space of preorders in terms of the existence of certain monotones relevant in various disciplines, which leads to a hierarchy of classes of preordered spaces. The conception of injective monotones then allows for a refinement of this hierarchy of preorders. 

Historically, much of the early development of real-valued representations has focused on \emph{total} preordered spaces that allow for the existence of utility functions. In particular, in the field of mathematical economics, \emph{utility theory} has pioneered the axiomatic study of conditions that ensure the existence of utility functions for a preordered set $(X,\preceq)$, where $X$ is a set of commodities and $\preceq$ is some total preference relation, a total preorder \cite{debreu1954representation,rebille2019continuous}.
Similarly, we can consider \emph{statistical estimation}, where the aim is to infer the distribution of a random variable $X$ from some of its realizations. Assuming the distribution belongs to a family $\{p_{\theta}\}_{\theta \in \mathbb{R}^N}$ for some $N>0$, a loss function $\ell:\mathbb{R}^N \rightarrow \mathbb{R}$ allows rating distributions according to how well they fit with the observed data: $p_{\theta} \preceq_{\ell} p_{\theta'}$ if and only if $-\ell(\theta) \leq -\ell(\theta')$ where $\theta, \theta' \in \mathbb{R}^N$ \cite{hennig2007some}. Choosing a loss function $\ell$ corresponds, thus, to defining a total preorder with a utility representation $\preceq_{\ell}$ on $\{p_{\theta}\}_{\theta \in \mathbb{R}^N}$. 

Another example of a preorder with a utility function is \emph{equilibrium thermodynamics}. Given a thermodynamic system, we consider $(X,\preceq_A)$ where $X$ is the set of all equilibrium states for the system and $x \preceq_A y$ if and only if $y$ is \emph{adiabatically accessible} from $x$ $\forall x,y \in X$ \cite{lieb1999physics}, that is, one can turn $x$ into $y$ using a device and a weight, with the device returning to its initial configuration at the end and the weight being allowed to change position in some gravitational field.
The main concern in the area is
the so-called \emph{entropy representation problem} \cite{candeal2001utility}, that is, the existence of a utility function, called entropy function, for $(X,\preceq_A)$ \cite{lieb1999physics}. 

Assuming a total preorder as in the previous examples is necessary for the existence of a utility function, but renders injective monotones uninteresting, as they become equivalent to strict monotones. When the totality assumption is dropped, the classes of preorders with these monotones can be distinguished. A well-known instance of non-total preorders with injective monotones is our running example, the uncertainty preorder. One of its relevant applications
lies in the study of quantum entanglement, as it characterizes the possible transformations using local operations and classical communications \cite[Theorem 1]{nielsen1999conditions}.
In physics, the uncertainty preorder given by majorization has recently also been extended. Given $\ell_1^1(\mathbb{R}^+):= \{ (p_i)_{i \in \mathbb{N}}| 0 \leq p_i \leq 1, \sum_{i=1}^{\infty}p_i=1\}$, we define \emph{infinite majorization} $\preceq_{IM}$ \cite{li2013neumann} for any $p,q \in \ell_1^1(\mathbb{R}^+)$ like
\begin{equation*}
    p  \preceq_{IM} q \ :\Leftrightarrow \ \sum_{i=1}^{k} p^{\downarrow}_i \leq \sum_{i=1}^{k} q^{\downarrow}_i \text{ } \forall k \in \mathbb{N},
\end{equation*}
where $p^{\downarrow}$ represents $p$ ordered in a decreasing way. Since $\preceq_{IM}$ is defined through a countable multi-utility, there exist injective monotones by Proposition \ref{countable implies injective monotone}. Finally, the uncertainty preorder is also an instance of multicriteria optimization \cite{ehrgott2005multicriteria}, also known as vector optimization \cite{jahn2009vector}, since it is concerned with the simultaneous optimization of a finite number of objective functions \eqref{uncert rela}. Notice strict and injective monotones  belong to the \textit{scalarization techniques} \cite{jahn2009vector,ehrgott2005multicriteria,bevilacqua2018multiobjective} in vector optimization and always exist, again by Proposition \ref{countable implies injective monotone}. 

Preordered spaces from the next general class, the ones with strict monotones, include \emph{general relativity}. Spacetime can be studied as a pair $(M, \preceq_C)$ where $M$ is a set of events and $\preceq_C$ is a causal relation, a partial order specifying which events can influence others, which lie to the future of others \cite{bombelli1987space}. A usual question is to establish sufficient conditions on $(M, \preceq_C)$ for the existence of strict monotones, which are referred to as \emph{time functions} \cite{minguzzi2010time} and 
are usually required to be continuous according to some topology. The study of physically plausible conditions from which countable multi-utilities or injective monotones can be constructed has, to our knowledge, not been addressed yet in the field.  Notice, spacetime was originally approached through a differentiable structure $(M,g)$, where $M$ is a manifold and $g$ a metric, and was only later studied as a partial order \cite{bombelli1987space}. 

Another case of preorders with strict monotones is  trumping. Consider $(\mathbb{P}_{\Omega}, \preceq_T )$ the space of probability distributions over some finite set $\Omega$, $\mathbb{P}_{\Omega}$, with the trumping preorder
     \begin{equation*}
        p \preceq_T q \iff \exists r \in \mathbb{P}_{\Omega'} \text{ } |\Omega'|< \infty \text{ s.t. } p \otimes r \preceq_M q \otimes r,
\end{equation*}
where $p \otimes r:= (p_1r_1,..,p_1r_{\Omega'},..,p_{\Omega}r_1,..,p_{\Omega} r_{\Omega'})$ $\forall p \in \mathbb{P}_{\Omega}$, $r \in \mathbb{P}_{\Omega'}$ and $|\Omega'|< \infty$ \cite{muller2016generalization}. Trumping extends majorization taking into account transformations using a third state, a \emph{catalyst}. As an example, consider $p:=(0.4,0.4,0.1,0.1)$, $q:=(0.5,0.25,0.25,0)$ and $r:=(0.6,0.4)$. Notice $\neg(p \preceq_M q)$ but $p \otimes r \preceq_M q \otimes r$, implying $p \preceq_T q$. Questions regarding physically meaningful strict monotones and multi-utilities for trumping are relevant \cite{turgut2007catalytic}. As no countable multi-utility has been found, it remains a question whether injective monotones do exist.

A final example from the most general class of preorders, the one where only monotones exist, are \emph{social welfare relations} (SWR) in economics. A SWR is a partial order $\preceq_{S}$ defined on the countably infinite product of the unit interval $X:=\prod_{n \in \mathbb{N}}\big[0,1\big]$. A SWR is said to be \textit{ethical} if $(1)$ given $x,y \in X$ with some $i,j \in \mathbb{N}$ such that $x_i=y_j$, $y_i=x_j$ and $x_k = y_k$ $\forall k \not \in \{i,j\}$ we have $x \sim_{S} y$ and $(2)$ given $x,y \in X$ where $x_i \leq y_i$ $\forall i \in \mathbb{N}$ and $x_j < y_j$ for some $j \in \mathbb{N}$ then $x \prec_S y$. Any ethical SWR is an example of a preordered space without strict monotones \cite[Proposition 1]{banerjee2010multi} and, thus, without both injective monotones and countable multi-utilities. 

\paragraph{Monotones and topology.} While we have focused on preordered spaces and left some brief comments regarding topology for Appendix \ref{topology}, in the past they have been often studied together. The original interest in functions representing order structures was concerned with (continuous) utility representations of total topological preordered spaces \cite{debreu1954representation,debreu1964continuity,eilenberg1941ordered}. Of particular importance were results concerning the existence of a continuous utility function for both connected and separable total topological preordered spaces \cite{eilenberg1941ordered} and for second countable total topological preordered spaces \cite{debreu1954representation}. Among the classical results we also find the existence of an order isomorphism between a subset of the real numbers and any total order with countably many jumps whose order topology is second countable \cite{fleischer1961numerical}. Based on the work of \cite{nachbin1965topology} relating topology and order theory, in particular a generalization of Urysohn's separation theorem, the classical results where reproved and sometimes generalized for example in \cite{herden1989existence,mehta1977topological,mehta1986existence,mehta1986theorem,mehta1988some,bosi2020topologies}. 

\paragraph{Multi-utility representations.} The study of non-total order structures was introduced in \cite{aumann1962utility}. Representation of non-total preorders by multi-utilities came later and was remarkably developed in \cite{evren2011multi}. Although strict monotones can be traced back to \cite{peleg1970utility,richter1966revealed}, there continue to be advances in the field \cite{herden2012utility,rebille2019continuous,bosi2020mathematical}. In fact, it was only recently in \cite{minguzzi2013normally} where strict monotone multi-utilities were introduced and later in \cite{alcantud2013representations,alcantud2016richter} where they were further studied. The relation of these ideas with optimization and the existence of maximal elements is also present in the literature \cite{bosi2017maximal,white1980notes,bevilacqua2018multiobjective,bevilacqua2018maximal,bosi2018upper}. Countable multi-utilities where studied particularly in \cite{bevilacqua2018representation,alcantud2016richter} while finite multi-utility representations were notably advanced in \cite{kaminski2007quasi,ok2002utility} and, in vector optimization, in \cite{jahn2009vector}.

\paragraph{Open questions.} 
While we have shown the existence of injective monotones for the widely studied class of preorders with countable multi-utilities, our construction is impractical since it relies on an infinite sum. For specific applications, injective monotones with a simpler representation are of interest. In general, any of the disciplines where these ideas are applied would benefit from a better understanding of the classification of preordered spaces in terms of real-valued monotones.
For example, regarding the maximum entropy principle, the classification could be useful in order to reconsider the reasoning behind the choice of Shannon entropy. Even though there have been many principled approaches to ``derive'' Shannon entropy as a measure of uncertainty in the past, such as \cite{Aczel1974,Shore1980}, and for many practical purposes its appealing properties overweigh the bias in choosing this particular strict monotone, the question remains whether one should \emph{maximize entropy} or \emph{maximize uncertainty}. Quantum physics could also benefit as, for instance, the preorder underlying entanglement catalysis, trumping, is not well understood \cite{muller2016generalization}. Many relevant open questions related to our work can also be found in \cite{bosi2020mathematical}, for example, while we have focused mostly on preordered spaces and made some remarks on semicontinuity, it would be important to study continuous injective monotones in terms of topological properties of the underlying spaces, as in the classical works on utility functions.


\newpage

\begin{appendix}
\section{Appendix}

\subsection{Entropy and the uncertainty preorder}
\label{majo}

In the following, we provide proofs for statements regarding the uncertainty preorder $\preceq_U$ and entropy $H$ that appear throughout the main part of this article, in particular, all results are stated with respect to the preordered space \smash{$(\mathbb P_\Omega, \preceq_U)$}, for a finite set $\Omega$.

\begin{lemma}[Basic facts] ~
\begin{enumerate}[label=(\roman*)]
    \item Shannon entropy is a strict monotone. If $|\Omega| \geq 3$ then it is not an injective monotone.
    \item If $|\Omega|\geq 3$ then for all $c\in (0,\log |\Omega|)$ there is an uncountable set $S_c$ such that $H(s)=c$ $\forall s \in S_c$. In particular, there are $p,q \in \mathbb P_\Omega$ with $c=H(p)=H(q)$ but $p\bowtie q$ for all $c\in (0,\log |\Omega|)$.
\end{enumerate}
\end{lemma}

\begin{proof}
$(i)$ Strict monotonicity of $H$ comes from the fact $H(p)= \sum_{i=1}^{|\Omega} f(p_i)$ where $f(x)=-x \log(x)$ is a strictly convex function. Given any other strictly convex $f$, strict monotonicity will still hold. One can find the details in \cite[C.1.a]{marshall1979inequalities}. $H$ is not an injective monotone for $|\Omega|\geq 3$ by $(ii)$.

$(ii)$ Given $p,q \in \mathbb{P}_{\Omega}$ we denote by $\overline{pq}$ the segment with endpoints $p,q$. Consider $u \in \mathbb{P}_{\Omega}$ the uniform distribution, $e_i,e_j \in \mathbb{P}_{\Omega}$ Dirac distributions for two different elements $i,j \in \Omega$ and some $c \in \big(0,\log|\Omega|\big)$. Consider some $c'$ s.t. $0<c'< \text{min} \{c,H(m)\}$ where $m$ is the middle point of $\overline{e_i e_j}$. By the intermediate value theorem there exists some $r \in \overline{e_i m}$ such that $H(r)=c'$. Consider now a parametrization of $\overline{e_i r}$: $\{r_t\}_{t \in [0,1]}$ and define $\ell_t\coloneqq \overline{r_t u}$ for each $t \in [0,1]$. Again by the intermediate value theorem, since $H(r_t)<c$ $\forall t \in [0,1]$, there exists some $p_t \in \ell_t$ such that $H(p_t)=c$ $\forall t \in [0,1]$. By construction, given $t,t'\in [0,1]$ $t\neq t'$ we have $p_t \neq p_{t'}$ since $\ell_t \cap \ell_{t'}= \{u\}$ whenever $t \neq t'$ which means $\{p_t\}_{t \in [0,1]}$ is uncountable. In particular, there are  $ t_c,t_c'\in [0,1]$ $t_c\neq t_c'$ such that $p_{t_c} \bowtie p_{t'_c}$ and $H(p_{t_c})=H(p_{t_c'})=c$ for every $c \in (0,\log|\Omega|)$.
\end{proof}

\smallskip

\begin{lemma}[Debreu separability] ~

\begin{enumerate}[label=(\roman*)]
\item If $|\Omega|=2$ then $(\mathbb{P}_{\Omega}, \preceq_U)$ is order separable. In particular, $(\mathbb{P}_{\Omega}, \preceq_U)$ is Debreu separable for $|\Omega|=2$.
\item If $|\Omega| \geq 3$ then any subset $Z \subseteq \mathbb{P}_{\Omega}$ which is Debreu dense in $(\mathbb{P}_{\Omega}, \preceq_U)$ has the cardinality of the continuum $|Z|= \mathfrak{c}$.
\item For any $|\Omega|<\infty$, there exists a countable upper dense set $Z \subseteq \mathbb{P}_{\Omega}$.
\end{enumerate}
 \end{lemma}
 
 \begin{proof} 
  For simplicity of notation, in the following we omit the subscript $U$ and thus write $\preceq$ for $\preceq_U$ (analogously for $\bowtie$ and $\prec$). 

 $(i)$ Consider $p,q \in \mathbb{P}$ such that $p \prec q$. By definition we have $q_1^{\downarrow} < p_1^{\downarrow}$. Consider some $s \in \mathbb{Q}$ such that $q_1^{\downarrow} < s < p_1^{\downarrow}$. Notice by normalization $\frac{1}{2} \leq q_1^{\downarrow} < s$ and by normalization again $1-s < s$. Thus $p \prec r \prec q$ where $r^{\downarrow}:=(s,1-s)$ and $\mathbb{Q}^2 \cap \mathbb{P}_{\Omega}$ is  countable and order dense in $(\mathbb{P}_{\Omega}, \preceq_U)$ for $|\Omega|=2$. In particular, $(\mathbb{P}_{\Omega}, \preceq)$ is Debreu separable for $|\Omega|=2$  which we could have known applying Theorem 1.4.8 in \cite{bridges2013representations} since for $|\Omega|=2$ there is a utility function, $u_1$.
 
 $(ii)$ Fix $|\Omega|=3$. Consider for some $x \in \big(\frac{1}{2}, 1 \big)$ some $p \in \mathbb{P}_{\Omega}$ such that $p_1^{\downarrow}=x$. Notice $p_2^{\downarrow} + p_3^{\downarrow} < x$ by normalization. Take some $\epsilon >0$ such that $p_2^{\downarrow}+\epsilon < x$ and $\epsilon \leq p_3^{\downarrow}$ to define $q \in \mathbb{P}_{\Omega}$ where $q^{\downarrow}:=(x,p_2^{\downarrow}+\epsilon,p_3^{\downarrow}-\epsilon)$. Notice $q \prec p$. Notice for any $x \in \big(\frac{1}{2}, 1 \big)$ we can define a pair $q_x,p_x\in \mathbb{P}_{\Omega}$ such that $q_x \prec p_x$ as we did before where for any $t \in \mathbb{P}_{\Omega}$ such that $q_x \preceq t \preceq p_x$ we have $t_1^{\downarrow} = x$. Given $Z \subseteq \mathbb{P}_{\Omega}$ a subset which is Debreu dense in $(\mathbb{P}_{\Omega}, \preceq)$ there exists for any $x \in \big(\frac{1}{2}, 1 \big)$ some $z_x \in Z$ such that $q_x \preceq z_x \preceq p_x$. Fix for every $x \in \big(\frac{1}{2}, 1 \big)$ some $z_x$. Notice, given $x,y \in \big(\frac{1}{2}, 1 )$, then $z_x = z_y$ implies $x=(z_x)_1^{\downarrow}=(z_y)_1^{\downarrow}=y$ which means that $\varphi: (\frac{1}{2}, 1)  \to Z, x \mapsto z_x$ is injective, implying $\mathfrak{c} \leq |Z|$. Since $Z \subseteq \mathbb{P}_{\Omega}$ and $|\mathbb{P}_{\Omega}| = \mathfrak{c}$ we have $|Z|=\mathfrak{c}$. In case $|\Omega| > 3$  any Debreu dense subset would also be Debreu dense in the subset with $|\Omega|=3$. We can thus follow the above lines and get the same conclusion for any $|\Omega| \geq 3$.
 
 $(iii)$ Consider $x,y \in \mathbb{P}_{\Omega}$ such that $x \bowtie y$. Since $x\bowtie y$ there exist $n,m \leq |\Omega|-1$ such that $\sum_{i=1}^{n}x_i^{\downarrow} < \sum_{i=1}^{n}y_i^{\downarrow}$ and $\sum_{i=1}^{m}x_i^{\downarrow} > \sum_{i=1}^{m}y_i^{\downarrow}$. Notice $y_i^{\downarrow} < 1$ $\forall i\leq |\Omega|$ since in the opposite case $y \preceq x$ $\forall x \in \mathbb{P}_{\Omega}$. Consider $1<k \leq |\Omega|$ the largest integer such that $y_k^{\downarrow} >0$ and define $\{\epsilon_i\}_{i=1}^{k-1}$ where 
 \begin{equation*}
    \begin{cases}
    0<\epsilon_i < \text{min} \{y_k^{\downarrow},\text{ } \sum_{j=1}^{m}x_j^{\downarrow} - \sum_{j=1}^{m}y_j^{\downarrow}\} & \text{if} \text{ } i=1,\\
    0<\epsilon_i < \text{min} \{ y_k^{\downarrow}- \sum_{j=1}^{i-1} \epsilon_j,\text{ } \sum_{j=1}^{m}x_j^{\downarrow} - \sum_{j=1}^{m}y_j^{\downarrow}-\sum_{j=1}^{i-1} \epsilon_j\} & \text{if} \text{ } 1<i\leq m,\\
   0<\epsilon_i < y_k^{\downarrow}- \sum_{j=1}^{i-1} \epsilon_j & \text{if} \text{ } m<i< k.
    \end{cases}
\end{equation*}
Notice $m<k$ since the opposite case leads to $\sum_{i=1}^{m}x_i^{\downarrow} > \sum_{i=1}^{m}y_i^{\downarrow}=\sum_{i=1}^{k}y_i^{\downarrow}=1$ contradicting normalization. For all $i < k$ choose $q_i \in \big(y_i^{\downarrow},y_i^{\downarrow} + \epsilon_i\big) \cap \mathbb{Q}$ such that $q_i \geq q_{i+1}$ and $q_k = 1- \sum_{i=1}^{k-1} q_i$. Then $z:=(q_1,q_2,..,q_k,0,..,0)$ has $|\Omega|-k$ zeros, the same number of zeros as $y$, and $z=z^{\downarrow}$, since $q_k < 1-\sum_{i=1}^{k-1} y_i^{\downarrow} = y_k \leq y_{k-1}<q_{k-1}$. By construction, we have $\sum_{j=1}^i z_j^{\downarrow} > \sum_{j=1}^i y_j^{\downarrow}$ $\forall i< k$ implying $\sum_{j=1}^{n}x_j^{\downarrow} < \sum_{j=1}^{n}z_j^{\downarrow}$, $\sum_{j=1}^{|\Omega|} z_j^{\downarrow}=\sum_{j=1}^k z_j^{\downarrow}=1$ and $\sum_{j=1}^m z_j^{\downarrow} < \sum_{j=1}^{m} x_j^{\downarrow}$ since 
 \begin{equation*}
     \sum_{j=1}^m z_j^{\downarrow} < \sum_{j=1}^m y_j^{\downarrow} + \epsilon_j < \sum_{j=1}^m y_j^{\downarrow} + \sum_{j=1}^m x_j^{\downarrow} - \sum_{j=1}^m y_j^{\downarrow} = \sum_{j=1}^m x_j^{\downarrow}
 \end{equation*}
 where in the first inequality we applied $z_j^{\downarrow} = q_j < y_j^{\downarrow}+\epsilon_j$ $\forall j \leq m$ and in the second we applied the definition of $\epsilon_m$ by which $ \sum_{j=1}^{m} \epsilon_j< \sum_{j=1}^m x_j^{\downarrow} - \sum_{j=1}^m y_j^{\downarrow}$. Thus, $x \bowtie z \prec y$. We have shown $\mathbb{Q}^{|\Omega|} \cap \mathbb{P}_{\Omega}$ is a countable upper dense set in $(\mathbb{P}_{\Omega}, \preceq_U)$ for any $|\Omega|< \infty$. 
 \end{proof}

\subsection{Proofs}

\subsubsection{Preorders without strict monotones \cite[Corollary 2.2]{alcantud2016richter}}
\label{no strict mono}
For example, consider the power set of the reals equipped with set inclusion, $(\mathcal{P}(\mathbb{R}), \subseteq)$. Since $\subseteq$ is reflexive, transitive, \emph{and} antisymmetric (i.e.~a partial order), by Szpilrajn extension theorem there exists a totally ordered space $(\overline{\mathcal{P}(\mathbb{R})}, \preceq)$ extending $(\mathcal{P}(\mathbb{R}), \subseteq)$, respecting the relations that already exist and relating the incomparable elements (e.g.~overlapping intervals). Hence, if there was a strict monotone $v:\overline{\mathcal{P}(\mathbb{R})} \rightarrow \mathbb{R}$, then $v(U) = v(V)$ for some $U, V\subseteq \mathbb R$ would imply that $U = V$, because w.l.o.g. $U\preceq V$, and $U\prec V$ cannot hold since $v$ is a strict monotone. This contradicts Cantor's theorem by which the cardinality of the power set $\mathcal P(\mathbb{R})$ is strictly greater than that of $\mathbb R$.

\subsubsection{Proof of Lemma \ref{series}}

First, note that for $ r \in (0, 1)$, we have
\begin{equation}
\label{ineq geo series}
  r \in \big(0, \tfrac{1}{2}\big) \iff r^m > \sum_{n=m+1}^{\infty} r^n \text{ } \forall m \geq 0 \, .
\end{equation}
This is a direct consequence of the closed-form formula of the geometric series and its partial sums, by which we have for any $r \in (0, 1)$ and $m\geq 0$, 
\[
\sum_{n=m+1}^\infty r^n = \frac{1}{1-r} - \sum_{n=0}^m r^n = \frac{r}{1-r} \, r^m \, ,
\]
so that $r^m> \sum_{n=m+1}^\infty r^n $ if and only if $r<1-r$, i.e.~$r\in (0,\frac{1}{2})$. 

Consider $x,y\in X$ and $m\in\mathbb N$ be the smallest index such that $\chi_{A_m}(x)\not=\chi_{A_m}(y)$. Assume $\chi_{A_m}(x)<\chi_{A_m}(y)$, in particular $\chi_{A_m}(x) = 0$ and $\chi_{A_m}(y) = 1$. Then
\begin{equation*}
\begin{split}
c(x) & \ \leq \  \sum_{n=0}^{m} r^n \chi_{A_n}(x) +  \sum_{n = m+1}^{\infty} r^{n} \\ 
& \ \stackrel{(a)}{<} \ \sum_{n=0}^{m} r^n \chi_{A_n}(x) +  r^m \ \stackrel{(b)}{=}  \ \sum_{n=0}^{m-1} r^n \chi_{A_n}(y) + r^m \chi_{A_m}(y)  \ \leq \ c(y)
\end{split}
\end{equation*}
where $(a)$ is due to \eqref{ineq geo series} and $(b)$ follows from the choice of $m$. Conversely, if $c(x)<c(y)$, then let $m$ be the first index where $\chi_{A_m}(x)\not = \chi_{A_m}(y)$. Clearly, if $\chi_{A_m}(y) <  \chi_{A_m}(x)$, then by the same argument as above, $c(y)<c(x)$, contradicting the hypothesis. Hence $\chi_{A_m}(x) <  \chi_{A_m}(y)$.

\subsubsection{Proof of Lemma \ref{lemma:sep}}

Given a family of increasing sets $(A_n)_{n \in \mathbb{N}}$, $c$ in \eqref{eq:constructInjectiveMonotone} is a monotone. If $x \prec y$, then by assumption there exists $n\in\mathbb N$ such that $x\not \in A_n$ and $y\in A_n$, i.e.~$\chi_{A_n}(x) < \chi_{A_n}(y)$, implying that $c$ is a strict monotone by Lemma \ref{series}. Furthermore, if $x \bowtie y$ there exists $n\in\mathbb N$ such that $x\not \in A_n$ and $y\in A_n$ or $y \not \in A_n$ and $x \in A_n$. There exists thus some $m \in  \mathbb{N}$ such that $\chi_{A_m}(x) \neq \chi_{A_m}(y)$ implying $c(x) \neq c(y)$ by Lemma \ref{series}. By Lemma \ref{basic characterizations} we have $c$ is an injective monotone.

\subsubsection{Proof of Proposition \ref{set charac monotones}}
\label{proof prop 7}

$(i)$ We only need to show given a strict monotone $v$ there exists a countable set $(A_n)_{n \in \mathbb{N}}$ that separates any pair $x,y \in X$ such that $x \prec y$, the converse is true by Lemma \ref{lemma:sep} Consider $A_n \coloneqq v^{-1}([q_n,\infty))$ where $(q_n)_{n \in \mathbb{N}}$ is a numeration of the rational numbers. If $x \prec y$ there exists $q_n \in \mathbb{Q}$ such that $v(x) < q_n < v(y)$ which means $y \in A_n$ and $x \not \in A_n$.

$(ii)$ Relying on Lemma \ref{lemma:sep} and $(i)$ we can take $(A_n)_{n \in \mathbb{N}}$ defined as in $(i)$ assuming $v$ is an injective monotone. Given $x,y \in X$ such that $x \bowtie y$ we have $v(x) \neq v(y)$ which implies there exists some $q_n \in \mathbb{Q}$ between $v(x)$ and $v(y)$ i.e. either $x \in A_n$ and $y \not \in A_n$ or $y \in A_n$ and $x \not \in A_n$.

\subsection{Semicontinuity}
\label{topology}

Much of the economic literature on utility representations in preordered spaces is concerned with topological questions, in particular, under which conditions on the preordered space one can expect that monotones and utilities satisfy certain continuity properties (e.g. \cite{debreu1964continuity,mehta1986existence,alcantud2016richter}). This is particularly important for optimization, since continuous functions attain their maximal elements  on compact sets. Therefore, in this section we collect the continuity properties of the injective monotones that appear in the main part of this article.    

Given a topology $\tau$, a triple $(X, \preceq, \tau)$ is called a \emph{preordered topological space}. A function $f: (X, \tau) \rightarrow (\mathbb{R}, \tau_{nat})$, where $\tau_{nat}$ is the tolopogy given by the Euclidean metric, is said to be \emph{upper semicontinuous} if $f^{-1}((-\infty, r)) \in \tau$ $\forall r \in \mathbb{R}$. 

Upper semicontinuous functions retain the property of continuous functions that they assume their maxima on compact sets, that is, they are effective on any compact set $B \subseteq X$.

Similarly, we say $(X, \preceq, \tau)$ is \emph{upper semicontinuous} if $i(x)=\{z \in X | x \preceq z\}$ is closed $\forall x \in X$. We may abuse notation and say that $\preceq$ is upper semicontinuous whenever $X$ and $\tau$ are clear.

\newpage
\begin{proposition}
\label{collection of continuity results}
Let $(X,\preceq, \tau)$ be a preordered topological space.
\begin{enumerate}[label=(\roman*)]
\item In Proposition \ref{strict mono equi r-p repre}, we can choose an upper semicontinuous multi-utility if the monotone is upper semicontinuous. 
\item In Proposition \ref{countable implies injective monotone}, we can add upper semicontinuous to both the hypothesis and the thesis.
\item In Proposition \ref{countable implies injective monotone multi}, if the preorder is upper semicontinuous then the equivalence remains true if upper semicontinuous is added to all clauses. 
\item In Proposition \ref{set charac monotones}, the monotones can be chosen to be upper semicontinuous if and only if the separating families consist of closed increasing sets.
\item In Proposition \ref{exists countable weakly upper dense implies equivalence R-P and injective monotone}, if the preorder is upper semicontinuous then the equivalence remains true if upper semicontinuous is added to all clauses. 
\end{enumerate}
\end{proposition}

Notice, the uncertainty preorder $\preceq_U$ is upper semicontinuous with respect to the Euclidean topology, since 
\begin{equation*}
i(p) = \{q\in \mathbb P_\Omega | p\preceq_u q\} = \bigcap_{i=1}^{|\Omega|-1} \{q \in \mathbb{P}_{\Omega} | u_i(p) \leq u_i(q)\}=\bigcap_{i=1}^{|\Omega|-1} u_i^{-1}( [u_i(p), \infty))
\end{equation*}
where $u_i^{-1}( [u_i(p), \infty))$ is closed, because all $u_i$ are upper semicontinuous.

\begin{proof}
$(i)$ If there exists an unpper semicontinuous injective monotone then we can construct w.l.o.g. an upper semicontinuous injective monotone $c:X \rightarrow (0,1)$. Since $\preceq$ is upper semicontinuous we know $\chi_{i(x)}$ is upper semicontinuous and, given the fact the class of upper semicontinuous functions is closed under addition by Proposition 1.5.12 in \cite{pedersen2012analysis},  $c_x$ in the proof of Proposition \ref{strict mono equi r-p repre} is upper semicontinuous $\forall x \in A_c$. Thus, $\{c\}\cup \{c_x\}_{x\in A_c}$ is an upper semicontinuous injective monotone multi-utility of $(X, \preceq, \tau)$.

$(ii)$ Take $(u_m)_{m \in M}$, $(A_n)_{n \in \mathbb{N}}$ and $c$ defined as in the proof of Proposition \ref{countable implies injective monotone}. If $(u_m)_{m \in M}$ is upper semicontinuous then $\forall n \in \mathbb{N}$ $A_n = u_{m_n}^{-1}([q_n, \infty[) \in \tau^c$ and $\chi_{A_n}(x)$ is upper semicontinuous $\forall n \in \mathbb{N}$. The class of upper semicontinuous function is closed under addition, product by positive scalars and uniform convergence by Proposition $1.5.12$ in \cite{pedersen2012analysis}. By the first two $c_N \coloneqq \sum_{n=0}^{N} 3^{-n} \chi_{A_n}$ is upper semicontinuous $\forall N \in \mathbb{N}$ and by the third $c= \lim_{N \to \infty} c_N$ is upper semicontinuous.

$(iii)$ Following $(ii)$ we get $\{c\}\cup\{c_{n_1,n_2}\}_{n_1<n_2}$ defined as in the proof of Propostion \ref{countable implies injective monotone multi} consists of upper semicontinuous injective monotones.

$(iv)$ Notice whenever $(A_n)_{n \in \mathbb{N}}$ in the proof of Propostion \ref{set charac monotones} is defined through an upper semicontinuous function, either a strict monotone or an injective montone, then $A_n$ is closed $\forall n \in \mathbb{N}$. Conversely, we can follow the proof of $(ii)$ to get upper semicontinuity for both a strict monotone and an injective monotone constructed as in Lemma \ref{lemma:sep}.

$(v)$ We again only show $(i)$ implies $(iii)$ in Proposition \ref{exists countable weakly upper dense implies equivalence R-P and injective monotone}. If $\preceq$ is upper semicontinuous then $\chi_{i(d)}$ is upper semicontinuous $\forall d \in D$ and since we can choose $u$ to be an upper semicontinuous strict monotone by hypothesis we get $\{u\} \bigcup \{\chi_{i(d)}\}$ is an upper semicontinuous countable multi-utility.
\end{proof}

\end{appendix}

\newpage

\bibliographystyle{plain}
\bibliography{references}

\end{document}